\documentclass{llncs}
\usepackage{amssymb}
\usepackage{amsmath}
\usepackage{xspace}
\usepackage{nicefrac}
\usepackage{comment}
\usepackage{tikz}
\usetikzlibrary{fit}

\usepackage[utf8]{inputenc}

\usepackage[noline, ruled]{algorithm2e}

\def\tsat{\textsc{3-SAT}\xspace}
\def\sat{\textsc{SAT}\xspace}
\def\is{\textsc{Max Independent Set}\xspace}

\def\sc{\textsc{Min Set Cover}\xspace}

\def\col{\textsc{Min Coloring}\xspace}

\def\ids{\textsc{Min Independent Dominating Set}\xspace}
\def\mmvc{\textsc{Max Minimal Vertex Cover}\xspace}

\def\atsp{\textsc{Min ATSP}\xspace}

\def\kipath{\textsc{$k$-Induced Path}\xspace}
\def\mim{\textsc{Max Induced Matching}\xspace}
\def\grundy{\textsc{Max Grundy Coloring}\xspace}
\def\mip{\textsc{Max Induced Path}\xspace}
\def\mif{\textsc{Max Induced Forest}\xspace}
\def\mit{\textsc{Max Induced Tree}\xspace}

\def\opt{\mathrm{opt}}

\let\le\leqslant
\let\ge\geqslant

\protect

\protect

\begin{document}

\title{Time-Approximation Trade-offs for Inapproximable Problems}

\author{\'Edouard Bonnet\inst{1}, Michael Lampis\inst{2}, Vangelis Th.
Paschos\inst{2,3}}

\institute{Hungarian Academy of Sciences, \email{bonnet.edouard@sztaki.mta.hu} \and LAMSADE, Université Paris
Dauphine, \email{michail.lampis@dauphine.fr,paschos@lamsade.dauphine.fr} \and Institut Universitaire de France}

\maketitle

\protect\thispagestyle{plain}

\begin{abstract}

In this paper we focus on problems which do not admit a constant-factor
approximation in polynomial time and explore how quickly their approximability
improves as the allowed running time is gradually increased from polynomial to
(sub-)exponential.  

We tackle a number of problems: For \ids, \mip, \textsc{Forest} and
\textsc{Tree}, for any $r(n)$, a simple, known scheme gives an approximation
ratio of $r$ in time roughly $r^{n/r}$. We show that, for most values of $r$,
if this running time could be significantly improved the ETH would fail.  For
\mmvc we give a non-trivial $\sqrt{r}$-approximation in time $2^{n/{r}}$.  We
match this with a similarly tight result.  We also give a $\log
r$-approximation for \atsp in time $2^{n/r}$ and an $r$-approximation for
\grundy in time $r^{n/r}$. 

Furthermore, we show that \sc exhibits a curious behavior in this
super-polynomial setting: for any $\delta>0$ it admits an
$m^\delta$-approximation, where $m$ is the number of sets, in just
quasi-polynomial time. We observe that if such ratios could be achieved in
polynomial time, the ETH or the Projection Games Conjecture would fail.

\end{abstract}

\section{Introduction}

One of the central questions in combinatorial optimization is how to deal
efficiently with NP-hard problems, with approximation algorithms being one of
the most widely accepted approaches. Unfortunately, for many optimization
problems, even approximation has turned out to be hard to achieve in polynomial
time.  This has naturally led to a more recent turn towards super-polynomial
and sub-exponential time approximation algorithms.  The goal of this paper is
to contribute to a systematization of this line of research, while adding new
positive and negative results for some well-known optimization problems.

For many of the most paradigmatic NP-hard optimization problems the best
polynomial-time approximation algorithm is known (under standard assumptions)
to be the trivial algorithm.  In the super-polynomial time domain, these
problems exhibit two distinct types of behavior.  On the one hand, APX-complete
problems, such as \textsc{MAX-3SAT}, have often been shown to display a ``sharp
jump'' in their approximability.  In other words, the only way to obtain any
improvement in the approximation ratios for such problems is to accept a fully
exponential running time, unless the Exponential Time Hypothesis (ETH) is false~\cite{MR10}.

A second, more interesting, type of behavior is displayed on the other hand by
problems which are traditionally thought to be ``very inapproximable'', such as
\textsc{Clique}.  For such problems it is sometimes possible to improve upon
the (bad) approximation ratios achievable in polynomial time with algorithms
running only in \emph{sub-exponential} time. In this paper, we concentrate on
such ``hard'' problems and begin to sketch out the spectrum of trade-offs
between time and approximation that can be achieved for them.

On the algorithmic side, the goal of this paper is to design
\emph{time-approximation trade-off schemes}.  By this, we mean an algorithm
which, when given an instance of size $n$ and an (arbitrary) approximation
ratio $r>1$ as a target, produces an $r$-approximate solution in time $T(n,r)$.
The question we want to answer is what is the best function $T(n,r)$, for each
particular value of $r$.  Put more abstractly, we want to sketch out, as
accurately as possible, the Pareto curve that describes the best possible
relation between worst-case approximation ratio and running time for each
particular problem. For several of the problems we examine the best known
trade-off algorithm is some simple variation of brute-force search in
appropriately sized sets. For some others, we present trade-off schemes with
much better performance, using ideas from exponential-time and parameterized
algorithms, as well as polynomial-time approximation.

Are the trade-off schemes we present optimal? A naive way to answer this
question could be to look at an extreme, already solved case: set $r$ to a
value that makes the running time polynomial and observe that the approximation
ratios of our algorithms generally match (or come close to) the best-known
polynomial-time approximation ratios.  However, this observation does not alone
imply satisfactorily the optimality of a trade-off scheme: it leaves open the
possibility that much better performance can be achieved when $r$ is restricted
to a different range of values.  Thus, the second, perhaps more interesting,
direction of this paper is to provide lower bound results (almost) matching
several of our algorithms \emph{for any point in the trade-off curve}.  For a
number of problems, these results show that the known schemes are (essentially)
the best possible algorithms, everywhere in the domain between polynomial and
exponential running time. We stress that we obtain these much stronger
\emph{sub-exponential inapproximability} results relying only on standard,
appropriately applied, PCP machinery, as well as the ETH.

\noindent\textbf{Previous work:} Moderately exponential and sub-exponential
approximation algorithms are relatively new topics, but most of the standard
graph problems have already been considered in the trade-off setting of this
paper.  For \is and \col an $r$-approximation in time~$c^{\nicefrac{n}{r}}$ was given by
Bourgeois et al.~\cite{BourgeoisEP09,BEP09b}.  For \sc, a $\log
r$-approximation in time~$c^{\nicefrac{n}{r}}$  and an $r$-approximation in time~$c^{\nicefrac{m}{r}}$,
where $n,m$ are the number of elements and sets respectively, were given by
Cygan, Kowalik and Wykurz~\cite{CyganKW09,BEP09}.  For \ids an $r$-approximation
in $c^{\nicefrac{n\log r}{r}}$ is given in~\cite{BCEP13}. An algorithm with similar
performance is given for \textsc{Bandwidth} in~\cite{CP10} and for
\textsc{Capacitated Dominating Set} in~\cite{CPW11}. In all the results above,~$c$ denotes some appropriate constant.

On the hardness side, the direct inspiration of this paper is the recent work
of Chalermsook, Laekhanukit and Nanongkai~\cite{ChalermsookLN13} where the
following was proved.
\begin{theorem} \label{thm:CLN}~\cite{ChalermsookLN13} For all $\varepsilon>0$,
for all sufficiently large $r=O(n^{\nicefrac{1}{2}-\varepsilon})$, if there
exists an $r$-approximation for \is running in
$2^{\nicefrac{n^{1-\varepsilon}}{r^{1+\varepsilon}}}$ then there exists a randomized
sub-exponential algorithm for \tsat.
\end{theorem}
Theorem~\ref{thm:CLN} essentially showed that the very simple approximation
scheme of~\cite{BourgeoisEP09} is probably ``optimal'', up to an arbitrarily
small constant in the second exponent, for a large range of values of $r$ (not
just for polynomial time).  The hardness results we present in this paper
follow the same spirit and in fact also rely on the technique of appropriately
combining PCP machinery with the ETH, as was done in~\cite{ChalermsookLN13}.
To the best of our knowledge, \is and \mim (for which similar results are given
in \cite{ChalermsookLN13}) are the only problems for which the trade-off curve
has been so accurately bounded.  The only other problem for which the
optimality of a trade-off scheme has been investigated is \sc.  For this
problem the work of Moshkovitz~\cite{M12} and Dinur and Steurer~\cite{DS14}
showed that there is a constant $c>0$ such that $\log r$-approximating \sc
requires time $2^{{(\nicefrac{n}{r})}^c}$. It is not yet known if this constant
$c$ can be brought arbitrarily close to 1.

\noindent\textbf{Summary of results:} In this paper we want to give upper and
lower bound results for trade-off schemes that match as well as the algorithm
of~\cite{BourgeoisEP09} and Theorem~\ref{thm:CLN} do for \is; we achieve this
for several problems.
%
%More specifically, we begin from the simple observation that for most graph
%problems where the solution is a set of vertices an $r$-approximation can be
%obtained in $c^{\nicefrac{n\log r}{r}}$. We then show the following:
\begin{itemize}
\item For \ids, there is no $r$-approximation in
$2^{\nicefrac{n^{1-\varepsilon}}{r^{1+\varepsilon}}}$ for any $r$, unless the
\emph{deterministic} ETH fails. This result is achieved with a direct reduction
from a quasi-linear PCP and is stronger than the corresponding result for \is
(Theorem~\ref{thm:CLN}) in that the reduction is deterministic and works for
all $r$.
\item For \mip, there is no $r$-approximation in $2^{o(\nicefrac{n}{r})}$ for any $r<n$,
unless the deterministic ETH fails. This is shown with a direct reduction from
\tsat, which gives a sharper running time lower bound.  For \mit and
\textsc{Forest} we show hardness results similar to Theorem~\ref{thm:CLN} by
reducing from \is.
\item For \mmvc we give a scheme that returns a $\sqrt{r}$-ap\-p\-ro\-x\-i\-ma\-ti\-on in
time $c^{\nicefrac{n}{r}}$, for any $r>1$. We complement this with a reduction from \is
which establishes that a $\sqrt{r}$-approximation in time
$2^{\nicefrac{n^{1-\varepsilon}}{r^{1+\varepsilon}}}$ (for any $r$) would disprove the randomized
ETH.
\item For \atsp we adapt the classical $\log n$-approximation into a $\log
r$-approximation in $c^{\nicefrac{n}{r}}$. For \grundy we give a simple $r$-approximation
in $c^{\nicefrac{n}{r}}$. For both problems membership in APX is still an open problem.
\item Finally, we consider \sc. Its approximability  in terms of $m$  is poorly
understood, even in polynomial time. With a simple refinement of an argument
given in~\cite{nelson07} we show how to obtain for any $\delta>0$ an
$m^\delta$-approximation in quasi-polynomial time
$2^{\log^{\nicefrac{(1-\delta)}{\delta}}n}$. We also observe that, if the ETH and the
Projection Games Conjecture~\cite{M12} are true, there exists $c>0$ such that
$m^c$-approximation cannot be achieved in polynomial time. This would imply
that the approximability of \sc changes dramatically from polynomial to
quasi-polynomial time. The only other problem which we know to exhibit this
behavior is \textsc{Graph Pricing}~\cite{ChalermsookLN13}.
%
% where the best result is a $\sqrt{m}$-approximation, and the best hardness
%$2^{\log^{1-o(1)} m}$~\cite{nelson07}. Intriguingly, the
%$\sqrt{m}$-approximation algorithm uses the $r$-approximation in $c^{m/r}$ time
%of~\cite{CyganKW09} as a sub-routine.  With a simple refinement of this
%algorithm's argument we show how to obtain for any $\delta>0$ an
%$m^\delta$-approximation in quasi-polynomial time $2^{\log^{(1-\delta)/\delta}
%n}$.  We also observe that, if the ETH and the Projection Games Conjecture~\cite{M12} are true, there exists $c>0$ such that $m^c$-approximation cannot be
%achieved in polynomial time. This would imply that \sc is a rare example of a
%problem where going from polynomial to quasi-polynomial time dramatically
%improves the best possible ratio. The only other problem which we know to
%exhibit this behavior is \textsc{Graph Pricing}~\cite{ChalermsookLN13}.
%
\end{itemize}

\section{Preliminaries and Baseline Results}

\subsubsection*{Algorithms}

In this paper we consider time-approximation trade-off schemes.  Such a scheme
is an algorithm that, given an input of size $n$ and a parameter~$r$, produces
an $r$-approximate solution (that is, a solution guaranteed to be at most a
factor $r$ away from optimal) in time~$T(n,r)$.  Sometimes we will overload
notation and allow trade-off schemes to have an approximation ratio that is
some other function of~$r$, if this makes the function~$T(n,r)$ simpler.  We
begin with an easy, generic, such scheme, that simply checks all subsets of a
certain size.
\begin{theorem} \label{thm:generic}
Let $\Pi$ be an optimization problem on graphs, for which the solution is a 
set of vertices and feasibility of a solution can be verified in polynomial time. 
Suppose that $\Pi$ satisfies one of the following sets of conditions:
\begin{enumerate}
\item The objective is $\mathrm{min}$ and some solution can be produced in polynomial time.
\item The objective is $\mathrm{max}$ and for any feasible solution $S$ there exists $u\in S$ such that
$S\setminus\{u\}$ is also feasible (weak monotonicity).
\end{enumerate} 
Then, for any $r>1$ (that may depend on the order $n$ of the input) there 
exists an $r$-approximation for $\Pi$ running in time $O^*((er)^{\nicefrac{n}{r}})$.
\end{theorem}

\begin{proof}
The algorithm simply tries all sets of vertices of size up to~$\nicefrac{n}{r}$. These are
at most $\nicefrac{n}{r} {n \choose r} =  O^*((er)^{\nicefrac{n}{r}})$. Each set is checked for
feasibility and the best feasible set is picked. In the case of minimization
problems, either we will find the optimal solution, or all solutions contain at
least~$\nicefrac{n}{r}$ vertices, so an arbitrary solution (which can be produced in
polynomial time) is an $r$-approximation. In the case of maximization, the weak
monotonicity condition ensures that there always exists a feasible solution of
size at most~$\nicefrac{n}{r}$.~\qed
\end{proof}

Because of Theorem~\ref{thm:generic}, we will treat this kind of qualitative
trade-off performance ($r$ approximation in time exponential in $\nicefrac{n\log r}{r}$) as
a ``baseline''. It is, however, not trivial if this performance can be achieved
for other types of graph problems (e.g. ordering problems).  Let us also note
that, for maximization problems that satisfy strong monotonicity (all subsets
of a feasible solution are feasible) the running time of Theorem~\ref{thm:generic} can be improved to $O^*(2^{\nicefrac{n}{r}})$~\cite{BourgeoisEP09}. 

\subsubsection*{Hardness}

The Exponential Time Hypothesis (ETH)~\cite{impagliazzo01} is the assumption
that there is no $2^{o(n)}$-algorithm that decides \tsat instances of size $n$.
All of our hardness results rely on the ETH or the (stronger) randomized ETH,
which states the same for randomized algorithms.

For most of our hardness results we also make use of known quasi-linear PCP
constructions. Such constructions reduce \tsat instances of size $n$ into CSPs
with size $n\log^{O(1)} n$, so that there is a gap between satisfiable and
unsatisfiable instances. Assuming the ETH, these constructions give a problem
that cannot be approximated in time $2^{o(\nicefrac{n}{\log^{O(1)}n)}}$ which we often
prefer to write as $2^{n^{1-\varepsilon}}$, though this makes the lower bound
slightly weaker. We note that, because of the poly-logarithmic factor added by
even the most efficient known PCPs, current techniques are often unable to
distinguish between whether the optimal running time for $r$-approximating a
problem is, say $2^{\nicefrac{n}{r}}$ or $r^{\nicefrac{n}{r}}$. The existence of linear PCPs, which at
the moment is open, could help further our understanding in this direction. To
make the sections of this paper more independent, we will cite the PCP theorems
we use as needed.

\section{\ids}

The result of this section is a reduction showing that for \ids, no trade-off
scheme can significantly beat the baseline performance of Theorem~\ref{thm:generic}, which qualitatively matches the best known scheme for this
problem~\cite{BCEP13}.  Thus, in a sense \ids is an ``inapproximable'' problem
in sub-exponential time.  Interestingly, \ids was among the first problems to
be shown to be inapproximable in both polynomial time~\cite{H93a} and FPT time~\cite{DFMR08}.

To show our hardness result, we will need an almost linear PCP construction
with perfect completeness.  Such a PCP was given by Dinur~\cite{dinur07}.
%, using\cite{bensasson06}.  
%We reproduce the lemma below.
\begin{lemma}[\cite{dinur07}, Lemma 8.3.]
There exist constants $c_1, c_2 > 0$ and a polynomial time reduction that transforms any SAT instance $\phi$ of size $n$ into a constraint graph $G= \langle (V,E), \Sigma, \mathcal C \rangle$ such that
\begin{itemize}
\item $|V|+|E| \leqslant n (\log n)^{c_1}$ and $\Sigma$ is of constant size.
\item If $\phi$ is satisfiable, then UNSAT$(G)=0$.
\item If $\phi$ is not satisfiable, then UNSAT$(G) \geqslant \nicefrac{1}{(\log n)^{c_2}}$.
\end{itemize}
\end{lemma}
Let us recall the relevant definitions from~\cite{dinur07}. A constraint graph
is a CSP whose variables are the vertices of $G$ and take values over $\Sigma$.
All constraints have arity 2 and correspond to the edges of $E$; with each
constraint $C_e$ we associate a set of satisfying assignments from $\Sigma^2$.
UNSAT$(G)$ is the fraction of unsatisfied constraints that correspond to the
optimal assignment to $V$. Observe that we only need here a PCP theorem where
UNSAT$(G)$ is at least inverse poly-logarithmic in $n$ (rather than constant).
The important property we need for our reduction is perfect completeness (that
is, UNSAT$(G)=0$ in the YES case).
\begin{theorem}\label{thm:mids-hardness}
Under ETH, for any $\varepsilon>0$ and $r \leqslant n$, an $r$-approximation for \ids cannot take time $O^*(2^{\nicefrac{n^{1-\varepsilon}}{r^{1+\varepsilon}}})$.
\end{theorem}
\begin{proof}
Let $G= \langle (V,E), \Sigma, \mathcal C \rangle$ the constraint graph obtained from any SAT formula $\phi$, applying the above lemma.
Let $s=|\Sigma|$, $n=|V|$ and $m=|E|$. 
We define an instance $G'=(V',E')$ of \ids in the following way.
For each vertex $v \in V$ and $a \in \Sigma$, we add a vertex $w_{v,a}$ in $V'$.
For each $v$, the $s$ vertices $w_{v,1}, w_{v,2}, \ldots, w_{v,s}$ are pairwise linked in $G'$ together with a dummy vertex $w_{v,0}$ and form a clique denoted by $C_v$. 
The idea would naturally be that taking $w_{v,a}$ in the independent dominating set corresponds to coloring $v$ by $a$.
For each edge $e=uv \in E$, and for each satisfying assignment $(i,j) \in C_e$ we add an independent set $I_{e,(i,j)}$ of $r$ vertices in $V'$, we link $w_{u,i}$ to all the vertices of the independent sets $I_{e,(i',j')}$ where $i' \in \Sigma \setminus \{i\}$ (and $j' \in \Sigma$), and we link $w_{v,j}$ to all the vertices of the independent sets $I_{e,(i',j')}$ where $(i',j) \in C_e$.
We finally add, for each edge $e=uv$, an independent set $I_e$ of $r$ vertices, and we link $w_{u,i}$ to all the vertices of $I_e$ if there is a pair $(i,j) \in C_e$ for some $j \in \Sigma$. 

If~$\phi$ is satisfiable, then UNSAT$(G)=0$, so there is a coloring $c:V
\rightarrow \Sigma$ satisfying all the edges.  Thus, $\bigcup_{v \in V}
\{w_{v,c(v)}\}$ is an independent dominating set of size~$n$.  It is
independent since there is no edge between~$w_{v,a}$ and~$w_{v',a'}$ whenever
$v \neq v'$.  It dominates~$\bigcup_{v \in V} C_v$ since one vertex is taken
per clique.  It also dominates~$I_e$ for every edge~$e$, by construction.  We
finally have to show that all the independent sets~$I_{uv,(i,j)}$ are
dominated.  If $c(u) \neq i$, then~$I_{uv,(i,j)}$ is dominated by~$w_{u,c(u)}$
(since $(c(u),c(v)) \in C_e$).  We now assume that $c(u) = i$.  Then~$I_{uv,(i,j)}$ is dominated by~$w_{v,c(v)}$, since $(c(u),c(v)) \in C_e$.

If~$\phi$ is not satisfiable, then UNSAT$(G) \geqslant \nicefrac{1}{(\log n)^{c_2}}$.
Any independent dominating set~$S$ has to take one vertex per clique~$C_v$ (to dominate the dummy vertex~$w_{v,0}$).
Let~$A$ be $S \cap \bigcup_{v \in V} C_v$, and let $c: V \rightarrow \Sigma$ be the coloring corresponding to~$A$.
Coloring~$c$ does not satisfy at least~$\nicefrac{m}{(\log n)^{c_2}}$ edges.
Let $E'' \subseteq E$ be the set of unsatisfied edges.
For each edge $e=uv \in E''$, let us show that at least one independent set of the form~$I_{uv,(i,j)}$ is not dominated by~$A$.
We may first observe that~$I_{uv,(i,j)}$ can only be dominated by~$w_{u,c(u)}$ or by~$w_{v,c(v)}$.
If there is no pair $(c(u),j') \in C_e$ for any~$j'$, then~$I_{e}$ is not dominated by construction.
If there is a pair $(c(u),j') \in C_e$ for some~$j'$, then~$I_{e,(c(u),j')}$ is not dominated by~$w_{u,c(u)}$ by construction, and is not dominated by~$w_{v,c(v)}$ since $(c(u),c(v)) \notin C_e$.

The only way of dominating those independent sets is to add to the solution all the vertices composing them, so a minimum independent dominating set is of size at least $n+\nicefrac{rm}{(\log n)^{c_2}} \geqslant \nicefrac{rn (\log n)^{c_1}}{(\log n)^{c_2}}=r'n$ setting $r'=\nicefrac{r (\log n)^{c_1}}{(\log n)^{c_2}}$. 

An $r'$-approximation for \ids can therefore decide the satisfiability of $\phi$.
The number of vertices in the instance of \ids is $n'=|V'| \leqslant (s+1)n+r(ms^2+1) \leqslant n(s+2+rs^2 (\log n)^{c_1})$.
So, for any $\varepsilon > 0$, if the $r'$-approximation algorithm for \ids runs in time $O^*(2^{\nicefrac{n'^{1-\varepsilon}}{r'^{1+\varepsilon}}})$, it contradicts ETH.~\qed
\end{proof}

\section{\mmvc}

In this section we deal with the \mmvc problem, which is the dual of \ids
(which is also known as \textsc{Minimum Maximal Independent Set}).
Interestingly, this turns out to be (so far) the only problem for which its
time-approximation trade-off curve can be well-determined, while being far from
the baseline performance of Theorem~\ref{thm:generic}. To show this result we
first present an approximation scheme that relies on a classic idea from
parameterized complexity: the exploitation of a small vertex cover.
\begin{theorem}\label{thm:mmvc}
For any $r$ such that $1<r\le \sqrt{n}$, \mmvc is $r$-approximable in time $O^*(2^{\nicefrac{n}{r^2}})$.
\end{theorem}
\begin{proof}
Our $r$-approximation algorithm begins by calculating a maximal matching
$M$ of the input graph.  If $|M|\ge \nicefrac{n}{r}$
then the algorithm simply outputs any arbitrary minimal vertex cover of $G$.  The solution, being a valid
vertex cover, must have size at least $|M|\ge \nicefrac{n}{r}$, and is therefore an
$r$-approximation.  

Otherwise, we partition the edges of $M$ into $r$ equal-sized groups arbitrarily. 
Let $V_i, 1\le i \le r$ be the set of vertices matched by the edges in group $i$.
By the bound on the size of $M$ we have that $|V_i|\le \nicefrac{2n}{r^2}$. 
We use $L$ to denote the set of vertices unmatched by $M$. 
Note that $L$ is of course an independent set.

The basic building block of our algorithm is a procedure which, given an
independent set $I$, builds a minimal vertex cover of $G$ that does not contain
any vertices of $I$. 
This can be done in polynomial time by first selecting $V\setminus I$ as a vertex cover of $G$, and then repeatedly removing from the cover redundant vertices one by one, until the solution is minimal. 
It is worthy of note here that this procedure guarantees the construction of a minimal vertex cover with size at least $|N(I)|$, where $N(I)$ is the set of vertices with a neighbor in $I$.

The algorithm now proceeds as follows: for each $i\in\{1,\ldots,r\}$ we iterate
through all sets $S\subset V_i$ such that $S$ is an independent set. 
For each such $S$ we initially build the set $S':= S \cup (L\setminus N(S))$. 
In words, we add to $S$ all its non-neighbors from $L$ to obtain $S'$, which is thus also an independent set. 
The algorithm then builds a minimal vertex cover of size at least $|N(S')|$ using the procedure of the previous paragraph. 
In the end we select the largest of the covers produced in this way.

The algorithm has the claimed running time.
The number of independent sets contained in $V_i$ is at most $2^{\nicefrac{n}{r^2}}$, since $G[V_i]$ has at most $\nicefrac{2n}{r^2}$ vertices and contains a perfect matching. 
Everything else takes polynomial time.

Let us therefore check the approximation ratio. 
Fix an optimal solution and let $R_i, i\in\{1,\ldots,r\}$ be the set of vertices of $V_i$ \emph{not} selected by this solution. 
Also, let $R_L$ be the vertices of $L$ not selected by the solution. 
Observe that $R:=R_L\cup \bigcup_{1\le i\le r} R_i$ is an independent set, and the solution has size $\opt=|N(R)|$, because all vertices of the solution must have an unselected neighbor.

Observe now that there must exist an $i\in\{1,\ldots,r\}$ such that $|N(R_i\cup R_L)| \ge \nicefrac{|N(R)|}{r}$. 
This is a consequence of the fact that for any two  sets $I_1,I_2$ such that $I_1\cup I_2$ is independent we have $N(I_1\cup I_2) = N(I_1)\cup N(I_2)$. 
Now, since the algorithm iterated through all independent sets in $V_i$, it must
have tried the set $S:=R_i$. 
From this it built the independent set $S':=R_i \cup (L\setminus N(R_i))$. 
Observe that $S' \supseteq R_i\cup R_L$, because $R_L$ does not contain any neighbors of $R_i$. 
It follows that $|N(S')| \ge |N(R_i\cup R_L)|$. Since the solution produced has size at least $|N(S')|$ we get the promised approximation ratio.~\qed
\end{proof}

The corresponding hardness result consists of a reduction from the \is instances constructed in Theorem~\ref{thm:CLN}.
\begin{theorem}\label{thm:mmvc-hardness}
Under randomized ETH, for any $\varepsilon>0$ and $r \leqslant n^{\nicefrac{1}{2}-\varepsilon}$, no $r$-approximation for \mmvc can take time $O^*(2^{\nicefrac{n^{1-\varepsilon}}{r^{2+\varepsilon}}})$.
\end{theorem}

Because we will need to rely on the structure of the
instances produced for Theorem~\ref{thm:CLN} in~\cite{ChalermsookLN13}, we
restate  here the relevant theorem:
\begin{theorem}[\cite{ChalermsookLN13}, Theorem 5.2.] For any sufficiently
small $\varepsilon > 0$ and any $r \leqslant n^{\nicefrac{1}{2}-\varepsilon}$,
there is a randomized polynomial reduction, which, from an instance of \sat
$\phi$ on $n$ variables, builds a graph $G$ with
$n^{1+\varepsilon}r^{1+\varepsilon}$ vertices such that with high probability:
\begin{itemize} \item If $\phi$ is a YES-instance, then $\alpha(G) \geqslant
n^{1+\varepsilon}r$.  \item If $\phi$ is a NO-instance, then $\alpha(G)
\leqslant n^{1+\varepsilon}r^{2\varepsilon}$.  
\end{itemize} 
\end{theorem}
\begin{proof}[Theorem \ref{thm:mmvc-hardness}] Let $\phi$ be any instance of
\sat and $G=(V,E)$ be the graph built from~$\phi$ with the reduction of Theorem~5.2. in~\cite{ChalermsookLN13}.  Keeping the same notation, we add~$\lceil r \rceil$ pendant vertices to each vertex of~$G$ and we call this new graph~$G'$.
The best solution for \mmvc in~$G'$ is to fix a maximum independent set~$I$ of~$G$ and to take the~$\lceil r \rceil$ pendant vertices to each vertices of~$I$, plus the vertices of $V \setminus I$.  This is true since~$\lceil r \rceil$ is
at least~$1$.  Let~$\opt$ be the size of a largest minimal vertex cover.

If $\phi$ is a YES-instance, then $\alpha(G) \geqslant n^{1+\varepsilon}r$, and $\opt > n^{1+\varepsilon}r^2$.
If $\phi$ is a NO-instance, then $\alpha(G) \leqslant n^{1+\varepsilon}r^{2\varepsilon}$, and $\opt < n^{1+\varepsilon}r^{1+2\varepsilon}+n^{1+\varepsilon}r^{1+\varepsilon} < 2n^{1+\varepsilon}r^{1+2\varepsilon}$.
Therefore, an approximation with ratio $r'=\nicefrac{r^{1-2\varepsilon}}{2}$ for \mmvc would permit to solve \sat.
Assuming ETH, this cannot take time $2^{o(n)}$.

As $n':=|V(G')|=n^{1+\varepsilon}r^{2+\varepsilon}$, such an approximation would not be possible in time $2^{\nicefrac{n'^{1-\varepsilon}}{r^{2+\varepsilon}}}$.
Renaming $r'$ by $r$ and $n'$ by $n$, an $r$-approximation would not be possible in time $O^*(2^{\nicefrac{n^{1-\varepsilon}}{r^{2+6\varepsilon}}})$.~\qed
\end{proof}

\section{\textsc{Induced Path, Tree and Forest}}

In this section we study the \mip, \textsc{Tree} and \textsc{Forest} problems,
where we are looking for the largest set of vertices inducing a graph of the
respective type.  These are all hard to approximate in polynomial time~\cite{K95,LY93}, and we observe that an easy reduction from \is shows that the
generic scheme of Theorem~\ref{thm:generic} is almost tight in sub-exponential
time for the latter two. However, the most interesting result of this section
is a direct reduction we present from \tsat to \mip. This reduction allows us
to establish inapproximability for this problem \emph{without} the PCP theorem,
thus eliminating the $\varepsilon$ from the running time lower bound.
\begin{theorem}\label{thm:forest} Under ETH, for any $\varepsilon>0$ and
sufficiently large $r \leqslant n^{\nicefrac{1}{2}-\varepsilon}$, an $r$-approximation for
\mif or \mit cannot take time
$2^{\nicefrac{n^{1-\varepsilon}}{(2r)^{1+\varepsilon}}}$.  
\end{theorem}
\begin{proof}
For \mif we simply observe that, if $\alpha(G)$ is the size of the largest
independent set of a graph, the largest induced forest has size between
$\alpha(G)$ (since an independent set is a forest) and $2\alpha(G)$ (since
forests are bipartite). The result then follows from Theorem~\ref{thm:CLN}.

For \mit, we repeat the same argument, after adding a universal vertex
connected to everything to the instances of \is of Theorem~\ref{thm:CLN}.~\qed
\end{proof}

\begin{theorem}\label{thm:kipath-hardness}
Under ETH, for any $\varepsilon>0$ and $r \leqslant n^{1-\varepsilon}$, an $r$-approximation for \kipath cannot take time $2^{o(\nicefrac{n}{r})}$.
\end{theorem}
\begin{proof}
Let~$\phi$ be any instance of \tsat.
For any positive integer~$r$, we build an instance graph~$G$ of \kipath in the following way.
For each clause~$C_i$ ($i \in [m]$) we add seven vertices $v^1_{i,1}, v^1_{i,2}, \ldots, v^1_{i,7}$ which form a clique~$C^1_i$ and correspond to the seven partial assignments of the three literals of~$C_i$ satisfying the clause (if there is only two literals, then there is only three vertices in the clique).
We add~$m$ vertices $v^1_1,v^1_2, \ldots, v^1_m$, and for all $i \in [2,m]$, we link~$v^1_i$ to all the vertices of the cliques~$C^1_{i-1}$ and all the vertices of the cliques~$C^1_i$.
Vertex~$v^1_1$ is only linked to all the vertices of~$C^1_1$.
The graph defined at this point is called~$H_1$.
We make $r-1$ copies of~$H_1$, denoted by $H_2$, \dots, $H_r$.
For each $j \in [2,r]$, the vertices of~$H_j$ are analogously denoted by $v^j_{i,1}, v^j_{i,2}, \ldots, v^j_{i,7}$ (vertices in the clique~$C^j_i$ corresponding to the clause~$C_i$) and~$v^j_i$.
For each $j \in [2,r]$, we link vertex~$v^j_1$ to all the vertices of the clique~$C^{j-1}_m$, and we add an edge between any two vertices corresponding to contradicting partial assignments, that is assignments attributing different truth values to the same variable (even if those vertices are in distinct~$H_i$s).
We call such an edge a \emph{contradicting edge}.
The edges within the cliques~$C^j_i$ can be seen as contradicting edges, but we will not call them so.

\begin{figure}\label{fig:kipath-reduction}
\centering
\begin{tikzpicture}

%H1

%C1
\node[preaction={fill, lightgray},draw,circle] (v1) at (-1.5,1.8) {};

\node[draw,rectangle,rounded corners] (v11) at (0,3.6) {$x_1x_2x_3$};
\node[preaction={fill, lightgray},draw,rectangle,rounded corners] (v12) at (0,3) {$x_1x_2\overline{x_3}$};
\node[draw,rectangle,rounded corners] (v13) at (0,2.4) {$x_1\overline{x_2}x_3$};
\node[draw,rectangle,rounded corners] (v14) at (0,1.8) {$x_1\overline{x_2}\overline{x_3}$};
\node[draw,rectangle,rounded corners] (v15) at (0,1.2) {$\overline{x_1}x_2x_3$};
\node[draw,rectangle,rounded corners] (v16) at (0,0.6) {$\overline{x_1}\overline{x_2}x_3$};
\node[draw,rectangle,rounded corners] (v17) at (0,0) {$\overline{x_1}\overline{x_2}\overline{x_3}$};

\node[draw,rectangle,rounded corners,fit=(v11) (v12) (v13) (v14) (v15) (v16) (v17)] (c1) {};

\draw (v1) -- (v11.west) ;
\draw[very thick] (v1) -- (v12.west) ;
\draw (v1) -- (v13.west) ;
\draw (v1) -- (v14.west) ;
\draw (v1) -- (v15.west) ;
\draw (v1) -- (v16.west) ;
\draw (v1) -- (v17.west) ;

%C2
\begin{scope}[xshift=3cm]
\node[preaction={fill, lightgray},draw,circle] (v2) at (-1.5,1.8) {};

\node[draw,rectangle,rounded corners] (v21) at (0,3.6) {$x_1x_2x_3$};
\node[preaction={fill, lightgray},draw,rectangle,rounded corners] (v22) at (0,3) {$x_1x_2\overline{x_3}$};
\node[draw,rectangle,rounded corners] (v23) at (0,2.4) {$x_1\overline{x_2}x_3$};
\node[draw,rectangle,rounded corners] (v24) at (0,1.8) {$x_1\overline{x_2}\overline{x_3}$};
\node[draw,rectangle,rounded corners] (v25) at (0,1.2) {$\overline{x_1}x_2x_3$};
\node[draw,rectangle,rounded corners] (v26) at (0,0.6) {$\overline{x_1}x_2\overline{x_3}$};
\node[draw,rectangle,rounded corners] (v27) at (0,0) {$\overline{x_1}\overline{x_2}\overline{x_3}$};

\node[draw,rectangle,rounded corners,fit=(v21) (v22) (v23) (v24) (v25) (v26) (v27)] (c2) {};

\draw (v2) -- (v21.west) ;
\draw[very thick] (v2) -- (v22.west) ;
\draw (v2) -- (v23.west) ;
\draw (v2) -- (v24.west) ;
\draw (v2) -- (v25.west) ;
\draw (v2) -- (v26.west) ;
\draw (v2) -- (v27.west) ;

\draw (v2) -- (v11.east) ;
\draw[very thick] (v2) -- (v12.east) ;
\draw (v2) -- (v13.east) ;
\draw (v2) -- (v14.east) ;
\draw (v2) -- (v15.east) ;
\draw (v2) -- (v16.east) ;
\draw (v2) -- (v17.east) ;
\end{scope}

%C3
\begin{scope}[xshift=6cm]
\node[preaction={fill, lightgray},draw,circle] (v3) at (-1.5,1.8) {};

\node[preaction={fill, lightgray},draw,rectangle,rounded corners] (v31) at (0,3.6) {$x_1x_2x_4$};
\node[draw,rectangle,rounded corners] (v32) at (0,3) {$x_1x_2\overline{x_4}$};
\node[draw,rectangle,rounded corners] (v33) at (0,2.4) {$x_1\overline{x_2}\overline{x_4}$};
\node[draw,rectangle,rounded corners] (v34) at (0,1.8) {$\overline{x_1}x_2x_4$};
\node[draw,rectangle,rounded corners] (v35) at (0,1.2) {$\overline{x_1}\overline{x_2}x_4$};
\node[draw,rectangle,rounded corners] (v36) at (0,0.6) {$\overline{x_1}x_2\overline{x_4}$};
\node[draw,rectangle,rounded corners] (v37) at (0,0) {$\overline{x_1}\overline{x_2}\overline{x_4}$};

\node[draw,rectangle,rounded corners,fit=(v31) (v32) (v33) (v34) (v35) (v36) (v37)] (c3) {};

\draw[very thick] (v3) -- (v31.west) ;
\draw (v3) -- (v32.west) ;
\draw (v3) -- (v33.west) ;
\draw (v3) -- (v34.west) ;
\draw (v3) -- (v35.west) ;
\draw (v3) -- (v36.west) ;
\draw (v3) -- (v37.west) ;

\draw (v3) -- (v21.east) ;
\draw[very thick] (v3) -- (v22.east) ;
\draw (v3) -- (v23.east) ;
\draw (v3) -- (v24.east) ;
\draw (v3) -- (v25.east) ;
\draw (v3) -- (v26.east) ;
\draw (v3) -- (v27.east) ;
\end{scope}

%C4
\begin{scope}[xshift=9cm]
\node[preaction={fill, lightgray},draw,circle] (v4) at (-1.5,1.8) {};

\node[draw,rectangle,rounded corners] (v41) at (0,3.6) {$x_2x_3x_4$};
\node[preaction={fill, lightgray},draw,rectangle,rounded corners] (v42) at (0,3) {$x_2\overline{x_3}x_4$};
\node[draw,rectangle,rounded corners] (v43) at (0,2.4) {$x_2x_3\overline{x_4}$};
\node[draw,rectangle,rounded corners] (v44) at (0,1.8) {$x_2\overline{x_3}\overline{x_4}$};
\node[draw,rectangle,rounded corners] (v45) at (0,1.2) {$\overline{x_2}x_3x_4$};
\node[draw,rectangle,rounded corners] (v46) at (0,0.6) {$\overline{x_2}\overline{x_3}x_4$};
\node[draw,rectangle,rounded corners] (v47) at (0,0) {$\overline{x_2}\overline{x_3}\overline{x_4}$};

\node[draw,rectangle,rounded corners,fit=(v41) (v42) (v43) (v44) (v45) (v46) (v47)] (c4) {};

\draw (v4) -- (v41.west) ;
\draw[very thick] (v4) -- (v42.west) ;
\draw (v4) -- (v43.west) ;
\draw (v4) -- (v44.west) ;
\draw (v4) -- (v45.west) ;
\draw (v4) -- (v46.west) ;
\draw (v4) -- (v47.west) ;

\draw[very thick] (v4) -- (v31.east) ;
\draw (v4) -- (v32.east) ;
\draw (v4) -- (v33.east) ;
\draw (v4) -- (v34.east) ;
\draw (v4) -- (v35.east) ;
\draw (v4) -- (v36.east) ;
\draw (v4) -- (v37.east) ;

\end{scope}

\end{tikzpicture}
\caption{The graph $H_1$ built for the instance $\{x_1 \lor \neg x_2 \lor x_3, x_1 \lor x_2 \lor \neg x_3, \neg x_1 \lor x_2 \lor \neg x_4, x_2 \lor \neg x_3 \lor x_4\}$.
$G$ is obtained by laying end to end $r$ copies of $H_1$.
The rectangle boxes are the cliques $C^j_i$, and the contradicting edges are not shown.
An induced path with $2m$ vertices is represented in gray and can be extended into one with $2rm$ vertices in $G$ (the formula being satisfiable).}
\end{figure}
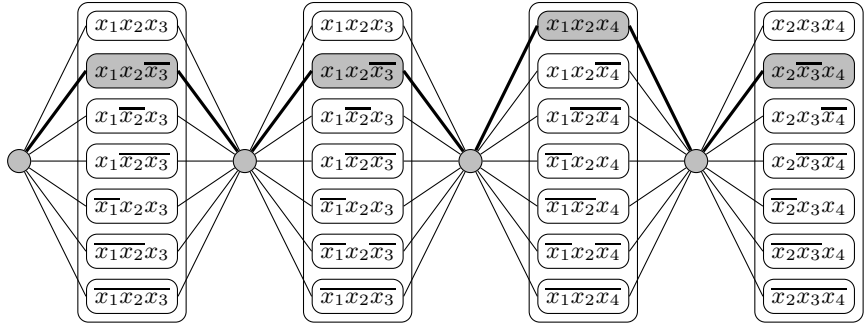

If~$\phi$ is satisfiable, let~$\tau$ be a truth assignment.  Let~$S$ be the set
of the~$rm$ vertices in cliques~$C^j_i$ agreeing with~$\tau$ (exactly one
vertex per clique).  The graph induced by $P=\bigcup_{1 \leqslant i \leqslant
m,1 \leqslant j \leqslant r} \{v^j_i\}$ $\cup S$ is a path with~$2rm$ vertices.
Indeed, $\forall i \in [2,m], j \in [r]$, the degree of~$v^j_i$ in~$G[P]$ is~$2$, since $|P \cap C^j_i|=1$ and $|P \cap C^j_{i-1}|=1$.  And, $\forall j \in [2,r]$, the degree of~$v^j_1$ in~$G[P]$ is~$2$, since $|P \cap C^j_1|=1$ and
$|P \cap C^{j-1}_m|=1$.  Vertex~$v^1_1$ has only degree~$1$ (one vertex in~$C^1_1$) and is one endpoint of the path.  The degree of the vertices of~$S$ in~$G[P]$ is also~$2$, since by construction there is no contradicting edge in the
graph induced by~$P$.  So, $\forall i \in [1,m-1], j \in [r]$, the only two
neighbors of the unique vertex in $S \cap C^j_i$ are~$v^j_i$ and~$v^j_{i+1}$.
And, $\forall j \in [r-1]$, the only two neighbors of the unique vertex in $S
\cap C^j_m$ are~$v^j_m$ and~$v^{j+1}_1$.  The degree in~$G[P]$ of the unique
vertex in $S \cap C^r_m$ is only~$1$; it is the other endpoint of the path. 

For each $i \in [m]$, we call \emph{column} $R_i$ the union of the $r$ cliques $C^1_i$, $C^2_i$, \dots, $C^r_i$.
Assume there is an induced path $G[Q]$ such that for some column $R_i$, $Q \cup R_i \geqslant 6$.
So there are at least four vertices $u_1, u_2, u_3, u_4$ which are in $Q \cup R_i$ and are not one of the two endpoints of $G[Q]$.
We set $U=\{u_1, u_2, u_3, u_4\}$.
We say that two vertices in the cliques $C^j_i$ \emph{agree} if they represent non contradicting (or \emph{compatible}) partial assignment.
We observe that two vertices in the same column $R_i$ agree iff they represent the same partial assignment.
First, we can show that all the vertices in $U$ have to (pairwise) agree.
If one vertex $u \in U$ does not agree with any of the other vertices in $U$, then $u$ has degree at least $3$ in $G[Q]$ (there are three contradicting edges linking $u$ to $U \setminus \{u\}$) which is not possible in a path.
So, any vertex in $U$ should agree with at least one vertex in $U \setminus \{u\}$.
The first possibility is that there are two pairs $(u,v)$ and $(w,x)$ of vertices spanning $U$, such that the vertices agree within their pair but the two pairs do not agree.
But that would create a cycle $uwvx$.
The only remaining possibility is that all the vertices in $U$ agree.
As those vertices are in the same column, they even represent the \emph{same} partial assignment.

Now, we will describe the path induced by~$Q$ by necessary conditions and derive that the formula is satisfiable.
Let~$u_5$ and~$u_6$ be two vertices in $(Q \cup R_i) \setminus U$, and $W=U \cup \{u_5,u_6\}$.
We observe that~$u_5$ and~$u_6$ should agree with the vertices of~$U$, otherwise their degree in~$G[Q]$ would be at least~$4$.
So, all the vertices in~$W$ (pairwise) agree.
The vertices of~$W$ are in pairwise distinct copies~$H_i$s.
Hence, there are at least~$4$ copies denoted by $H_{a_1}, H_{a_2}, H_{a_3}, H_{a_4}$ which contain a vertex of~$W$ and \emph{do not} contain an endpoint of~$G[Q]$.
Let~$v^{a_1}_{i,h}$ be the unique vertex in $W \cap H_{a_1}$.
By the previous remarks, $\forall p \in \{2,3,4\}$,~$v^{a_p}_{i,h}$ is the unique vertex in $W \cap H_{a_p}$.
For each $p \in [4]$, the two neighbors of~$v^{a_p}_{i,h}$ in~$G[Q]$ have to be~$v^{a_p}_i$ and~$v^{a_p}_{i+1}$.
Vertex~$v^{a_p}_{i,h}$ cannot incident to a contradicting edge, otherwise it would create a vertex of degree at least~$4$ in the path.
At its turn, vertex~$v^{a_p}_{i+1}$ has degree~$2$ in~$G[Q]$, and its second neighbor has to be in the clique~$C^{a_p}_{i+1}$ (if its second neighbor was also in~$C^{a_p}_i$, it would form a triangle).
Let~$w_{p,i+1}$ be the unique vertex in $C^{a_p}_{i+1} \cap P$.
By the same arguments as before,~$w_{1,i+1}$, $w_{2,i+1}$, $w_{3,i+1}$, and~$w_{4,i+1}$ should all agree.
This way we can extend the four fragments of paths to column~$R_{i+1}$ up to~$R_m$.
Symmetrically, we can extend the fragments of paths to column~$R_{i-1}$ to~$R_1$. 
Now, if we just consider the path induced by $Q \cup H_{a_1}$, it goes through consistent partial assignments for each clause of the instance.
The global assignment, built from all those partial assignments, satisfies all the clauses.
So, the contrapositive is, if~$\phi$ is not satisfiable, then for all $i \in [m]$, $|R_i \cup Q|<6$.
This implies $|Q|<10m$.

The number of vertices of $G$ is $8rm$.
Recall that, under ETH~\cite{impagliazzo01}, \tsat is not solvable in $2^{o(m)}$.
Thus, under ETH, any $r$-approximation for \kipath cannot take time
$2^{o(\nicefrac{n}{r})}$.~\qed
%  We can observe that this result also holds for \kicycle by
%linking all the vertices of $C^r_m$ to $v_1$.  
\end{proof}

\section{\atsp and Grundy Coloring}

In this section we deal with two problems for which the best known hardness of
approximation bounds are small constants~\cite{KLS13,kortsarz07}, but no
constant-factor approximation is known. We thus only present some algorithmic
results.

For \atsp, the version of the TSP where we have the triangle inequality but
distances may be asymmetric, the best known approximation algorithm has ratio
$O(\nicefrac{\log n}{\log\log n})$~\cite{AGMGS10}. Here, we show that a classical, simpler $\log
n$-approximation~\cite{FriezeGM82} can be adapted into an approximation scheme matching its
performance in polynomial time. Whether the same can be done for the more
recent, improved, algorithm remains as an interesting question.
\begin{theorem}\label{thm:atsp}
For any $r \leqslant n$, \atsp is $\log r$-approximable in time $O^*(2^{\nicefrac{n}{r}})$.
\end{theorem}
\begin{proof} We roughly recall the $\log
n$-approximation of \atsp detailed in~\cite{FriezeGM82}.  The idea is to solve
the problem of finding a (vertex-)disjoint union of circuits spanning the graph
with minimum weight.  This can be expressed as a linear program and therefore
it can be solved in polynomial time.  Let the circuits be $C_1,C_2, \ldots
C_h$.  We observe that the total length of the circuits is bounded by $\opt$
the optimum value for \atsp.  We choose arbitrarily a vertex $v_i$ in each
$C_i$ and recurse on the graph induced by $\{v_1,v_2,\ldots,v_h\}$.  By the
triangle inequality, we can combine a solution of \atsp in
$G[\{v_1,v_2,\ldots,v_h\}]$ to the circuits $C_i$s, and get a solution whose
value is bounded by the sum of the lengths of the $C_i$s plus the value of the
solution for $G[\{v_1,v_2,\ldots,v_h\}]$, which would be $2\opt$ if we solve
$G[\{v_1,v_2,\ldots,v_h\}]$ to the optimum.  In general, the depth of recursion
is a bound on the ratio (see~\cite{FriezeGM82}).  At each recursion step, the
number of vertices in the remaining graph is at least divided by two.  So,
after at most $\log n$ recursions the algorithm terminates, hence the ratio.

Now, we can afford some superpolynomial computations. 
After $\log r$ recursions the number of vertices in the remaining graph is no more than $\nicefrac{n}{2^{\log r}}=\nicefrac{n}{r}$.
We solve optimally this instance by dynamic programming in time
$O^*(2^{\nicefrac{n}{r}})$.  The solution that we output has length smaller than
$\log r \cdot \opt$.~\qed
 \end{proof}

\grundy is the problem of ordering the vertices of a graph so that a greedy
first-fit coloring applied on that order would use as many colors as possible.
 Unless NP$\subseteq$RP, \grundy admits no PTAS~\cite{kortsarz07}, but it is unknown if it can be $o(n)$-approximated.  
%\grundy
%is solvable in time $O^*(2.246^n)$~\cite{bonnet14}.

Observe that, since this is not a subgraph problem, it is not \emph{a priori} obvious
that the baseline trade-off performance of Theorem \ref{thm:generic} can be
achieved. However, we give a simple trade-off scheme that does exactly that by
reducing the ordering problem to that of finding an appropriate ``witness'',
which is a set of vertices.
\begin{theorem}\label{thm:grundy} For any $r>1$, \grundy can be
$r$-approximated in time $O^*(c^{\nicefrac{n \log r}{r}})$, for some constant $c$.
\end{theorem}
\begin{proof} Let $G=(V,E)$ be any instance of
\grundy, and~$r$ any real value.  Here, we call \emph{minimal witness} of~$G$
achieving color~$k$, an induced subgraph~$W$ of~$G$ whose grundy number is~$k$,
such that all the induced subgraphs of~$W$ different from~$W$ have strictly
smaller grundy numbers.

Let $k$ be the grundy number of $G$ and $W$ be a minimal witness.
Let $C_1 \uplus C_2 \uplus \ldots \uplus C_k$ be a partition of $V(W)$ corresponding to the color classes in an optimal coloring.
Let $A_1,A_2, \ldots, A_{\lfloor \nicefrac{k}{r} \rfloor}$ be the $\lfloor \nicefrac{k}{r} \rfloor$ smallest (in terms of number of vertices) color classes among the $C_i$s.
Let $S=A_1 \uplus A_2 \uplus \ldots \uplus A_{\lfloor \nicefrac{k}{r} \rfloor}$.
Obviously $|V(W)| \leqslant n$, so $|S| \leqslant \nicefrac{n}{r}$.

The algorithm exhausts all the subset of $\nicefrac{n}{r}$ vertices.
For each subset of vertices, we run the exact algorithm running in time $O^*(2.246^n)$ on the corresponding induced subgraph.
Thus, the algorithm takes time $O^*(2^{\nicefrac{n \log r}{r}} 2.246^{\nicefrac{n}{r}})$.
As $|S| \leqslant \nicefrac{n}{r}$, the algorithm considers at some point $S$ or a superset of $S$.
We just have to show that the optimal grundy coloring of $S$ is an $r$-approximation.
Let us re-index the $A_j$s by increasing values of their index in the $C_i$s,
say $B_1,B_2, \ldots, B_{\lfloor \nicefrac{k}{r} \rfloor}$.  Then for each $i \in
[1,\lfloor \nicefrac{k}{r} \rfloor]$, we can color $B_i$ with color $i$ and achieve
color $\lfloor \nicefrac{k}{r} \rfloor$.~\qed 
\end{proof}

\section{Set Cover}

In this section we focus on the classical \sc problem, on inputs with $n$
elements and $m$ sets. In terms of $n$, a $\log r$-approximation is known in
time roughly $2^{\nicefrac{n}{r}}$. Moshkovitz~\cite{M12} gave a reduction from
$N$-variable \tsat which, for any $\alpha<1$ produces instances with universe
size $n = N^{O(\nicefrac{1}{\alpha})}$ and gap $(1-\alpha)\ln n$. Setting $\alpha =
\nicefrac{\ln(\nicefrac{n}{r})}{\ln n}$ translates this result to the terminology of our paper, and
shows a running time lower bound of $2^{(\nicefrac{n}{r})^c}$, for some $c>0$. Thus, even
though the picture for this problem is not as clear as for, say \is, it appears
likely that the known trade-off scheme is optimal.

We consider here the complexity of the problem as a function of $m$. This is a
well-motivated case, since for many applications $m$ is much smaller than $n$~\cite{nelson07}. Eventually, we would like to investigate whether the
known $r$-approximation in time $2^{\nicefrac{m}{r}}$ can be improved. Though we do not
resolve this question, we show that the approximability status of this problem
is somewhat unusual.

In polynomial time, the best known approximation algorithm has a guarantee of
$\sqrt{m}$~\cite{nelson07}. We first observe that the simple argument of this
algorithm can be extended to quasi-polynomial time.
\begin{theorem}\label{thm:scalc} For any $\delta>0$ there is an
$m^\delta$-approximation algorithm for \sc running in time $O^*(c^{(\log
n)^{\nicefrac{(1-\delta)}{\delta}}})$.
\end{theorem}
\begin{proof}
The argument is similar to that of~\cite{nelson07}. We distinguish two cases:
if $m^\delta > \ln n$, then we can run the greedy polynomial time algorithm and
return a solution with ratio better than $m^\delta$. So assume that $m^\delta <
\ln n$.

Now, run the $r$-approximation of~\cite{CyganKW09}, setting $r = m^\delta$. The
running time is (roughly) $2^{\nicefrac{m}{r}} = 2^{m^{1-\delta}}$. The result follows
since $m < (\ln n)^{\nicefrac{1}{\delta}}$.~\qed
\end{proof}
The above result is somewhat curious, since it implies that in quasi-polynomial
time one can obtain an approximation ratio better than that of the best known
polynomial-time algorithm. This leaves open two possibilities: either
$\sqrt{m}$ is not in fact the optimal ratio in polynomial time, or there is a
jump in the approximability of \sc from polynomial to quasi-polynomial time. We
remark that, though this is rare, there is in fact another problem which
displays exactly this behavior: for \textsc{Graph Pricing} the best
polynomial-time ratio is $\sqrt{n}$, while $n^\delta$ can be achieved in time
$O^*(c^{(\log m)^{\nicefrac{(1-\delta)}{\delta}}})$~\cite{ChalermsookLN13}. 

We do not settle this question, but observe that a combination of known
reductions for \sc, the ETH and the Projection Games Conjecture of~\cite{M12}
imply that the optimal ratio in polynomial time is $m^c$ for some 
$c>0$.  Thus, \sc is indeed likely to behave in a way similar to
\textsc{Graph Pricing}.  For Theorem~\ref{thm:schard} we essentially reuse the 
combination of reductions used in~\cite{CHK13} to obtain FPT
inapproximability results for \sc.
\begin{theorem} \label{thm:schard} Assume the ETH and the PGC. Then, there
exists a $c>0$ such that there is no $m^c$-approximation for \sc running in
polynomial time.
\end{theorem}
\begin{proof}
As mentioned, the proof reuses the reduction of~\cite{CHK13}, which in turn
relies on the ETH, the PGC and classical reductions for \sc. To keep the
presentation as short and self-contained as possible we simply recall Theorem 5
of~\cite{CHK13}, without giving a detailed proof (or a definition of the PGC).
\begin{theorem}~\cite{CHK13} If the Projection Games Conjecture holds, for any
$r>1$ there exists a reduction from \tsat of size $N$ to \sc with the following
properties:
\begin{itemize}
\item YES instances produce \sc instances where the optimal cover has size
$\beta$, NO instances produce \sc instances where the optimal cover has size at
least $r\beta$.
\item The size $n$ of the universe is $2^{O(r)}\mathrm{poly}(N,r)$.
\item The number of sets $m$ is $\mathrm{poly}(N)\cdot \mathrm{poly}(r)$.
\item The reduction runs in time polynomial in $n,m$.
\end{itemize}
\end{theorem}
Using the above reduction, we can conclude that there exists \emph{some}
constant $c$ such that $m^c$-approximation for \sc is impossible in polynomial
time, under the ETH. The constant $c$ depends on the hidden exponents of the
polynomials of the above reduction. The way to do this is to set $r$ to be some
polynomial of $N$, say $r=\sqrt{N}$. Then, the reduction runs in time
sub-exponential in $N$ (roughly $2^{\sqrt{N}}$) and produces a gap that is
polynomially related to $m$. If in polynomial time we could $r$-approximate the
new instance, this would give a sub-exponential time algorithm for \tsat.~\qed

\end{proof}

\bibliographystyle{abbrv}
\bibliography{subexp}

\end{document}